\documentclass[11pt]{article}

\usepackage[margin=1in]{geometry}
\usepackage{parskip}
\usepackage{hyperref}
\usepackage{subfig,graphicx}
\usepackage{xcolor}
\usepackage{tikz}
\usetikzlibrary{positioning}
\tikzset{n/.style = {draw, circle}}
\tikzset{e/.style = {fill=white}}

\usepackage{amsmath,amssymb,amsthm}
\usepackage{mathrsfs}
\newtheorem{problem}{Problem}
\newtheorem{theorem}{Theorem}
\newtheorem{definition}{Definition}

\DeclareMathOperator*{\argmin}{arg\,min}
\DeclareMathOperator*{\argmax}{arg\,max}

\newcommand{\ie}{\emph{i.e.}}
\newcommand{\eg}{\emph{e.g.}}
\newcommand{\cC}{\mathcal{C}}
\newcommand{\sC}{\mathscr{C}}
\newcommand{\cS}{\mathcal{S}}
\newcommand{\cT}{\mathcal{T}}
\newcommand{\cA}{\mathcal{A}}

\title{An alignment problem}
\author{Emma L. McDaniel, Armin R. Mikler, Chetan Tiwari, and Murray Patterson}
\date{Georgia State University}

\begin{document}

\maketitle

\begin{abstract}
  This work concerns an alignment problem that has applications in
  many geospatial problems such as resource allocation and building
  reliable disease maps.  Here, we introduce the problem of optimally
  aligning $k$ collections of $m$ spatial supports over $n$ spatial
  units in a $d$-dimensional Euclidean space.  We show that the
  1-dimensional case is solvable in time polynomial in $k$, $m$ and
  $n$.  We then show that the 2-dimensional case is NP-hard for 2
  collections of 2 supports.  Finally, we devise a heuristic for
  aligning a set of collections in the 2-dimensional case.
\end{abstract}

\section{Introduction}
\label{sec:intro}

This work concerns the problem of aligning collections of spatial
supports which share a common set of spatial units.  For example,
Figures~\ref{fig:eg-a} and~\ref{fig:eg-b} each depicts a collection of
four supports (green, yellow, orange, and blue) which shares a common
set of 16 spatial units (rectangular blocks).  The goal is to swap
units from one support to another within each collection (change the
colors of blocks) until the collections are identical, \ie, are
aligned, as depicted in Figure~\ref{fig:eg-c}.  Note that there are
many different ways to align the supports, \ie, the alignment depicted
in Figure~\ref{fig:eg-c} is not unique.  With this in mind, it would
be preferable to align such collections using the minimum number of
(possibly weighted) swaps.  This optimization problem is easy in some
cases, and (NP-) hard in others.

In general, the alignment problem is on a set $U$ of $n$ \emph{spatial
  units}, each unit $u \in U$ having a population count $p(u)$ within
a certain spatial boundary, disjoint from other units.  These spatial
units can represent census tracts or ZIP code tabulation areas.
Constructing maps, \eg, choropleth maps, which reflect certain rates
within a population, such as cancer incidence, provides an intuitive
way to portray the geospatial patterns of such rates.  This can
provide decision support in public health surveillance, which can aid
officials to form the appropriate policies.
Building such a map at the level of an individual unit can produce
misleading results due to small populations in some units, resulting
in statistically unstable rates.  To remedy this issue, sets $s
\subseteq U$ of contiguous units are aggregated to create larger
\emph{spatial supports} with adequate population counts to ensure a
stable rate calculation, as depicted in Figure~\ref{fig:eg-a}.

\begin{figure}[!ht]
  \centering
  \subfloat[][]{\includegraphics[width=.33\textwidth]{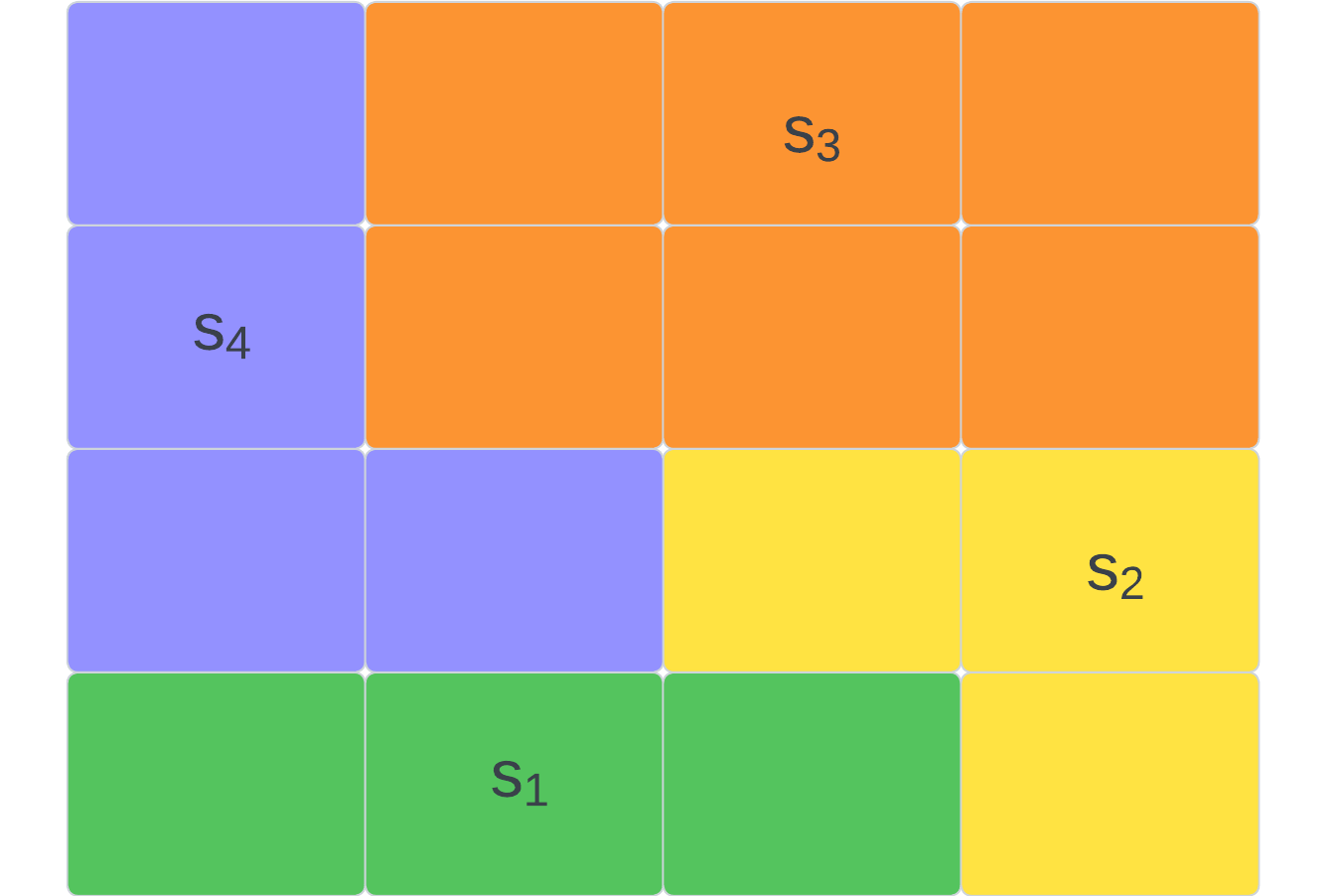}\label{fig:eg-a}}
  \subfloat[][]{\includegraphics[width=.33\textwidth]{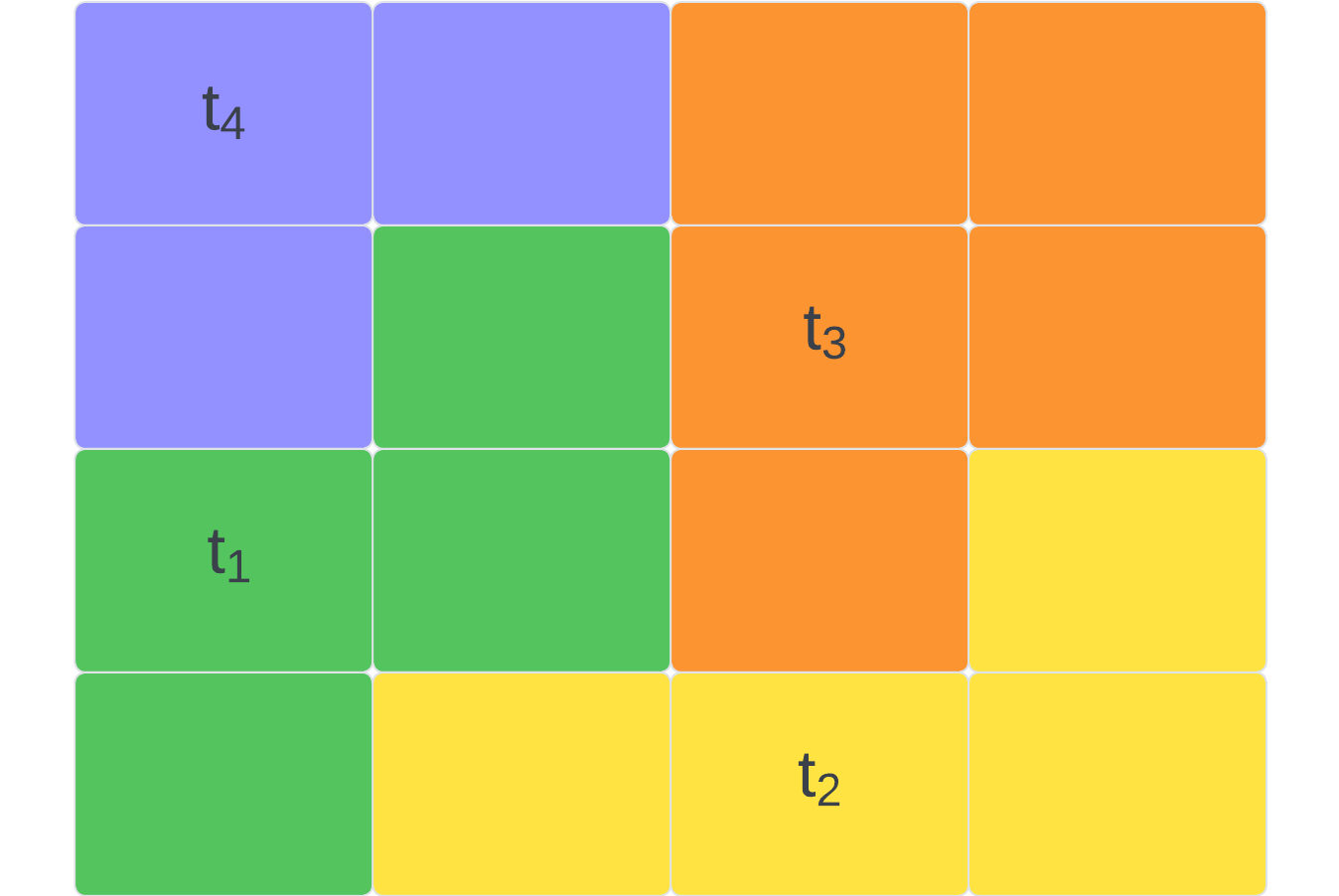}\label{fig:eg-b}}
  \hfil
  \subfloat[][]{\includegraphics[width=.294\textwidth]{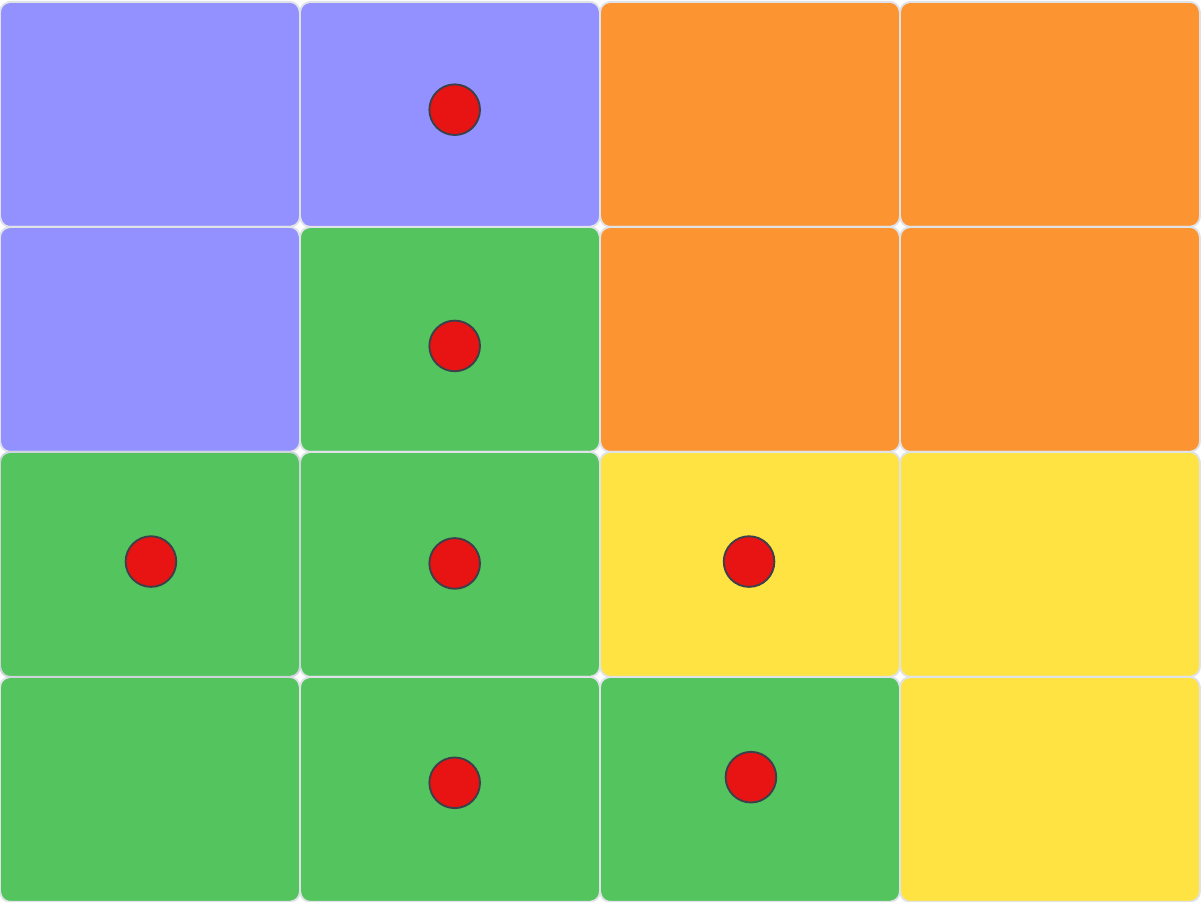}\label{fig:eg-c}}
  \caption{Collections (a) $\cS = \{s_1, s_2, s_3, s_4\}$ and (b) $\cT
    = \{t_1, t_2, t_3, t_4\}$ of spatial supports over the set $U =
    \{u_{1,1}, u_{1,2}, \dots, u_{4,4}\}$ of 16 spatial units.  An
    alignment (c) of $\cS$ and $\cT$.  The red dots mark the units
    ($u_{2,1}, u_{3,1}, u_{1,2}, u_{2,2}, u_{3,2}, u_{2,3}, u_{2,4}$)
    on which $\cS$ and $\cT$ disagree.}
  \label{fig:eg}
\end{figure}

Suppose a certain rate, \eg, prostate cancer incidence, can be mostly
explained by a factor such as age.  In this case, we want to create
several maps, which represent each age stratum in order to more
clearly portray this factor in determining such rate.  For the sake of
illustration, suppose that $\cS$ and $\cT$, depicted in
Figures~\ref{fig:eg-a} and~\ref{fig:eg-b} are two such maps,
represented as collections of supports over $U$.  In $\cS$, the
populations $p_\cS(u)$ of each unit $u$ in $s_1, s_2, s_3, s_4$ are
20, 20, 10, 15, respectively---\eg, $p_\cS(u_{1,1}) = 20$ ($u_{1,1}
\in s_1$), and $p_\cS(u_{4,4}) = 10$ ($u_{4,4} \in s_3$)---while the
populations $p_\cT(u)$ of each unit $u$ in $t_1, t_2, t_3, t_4$ are
15, 15, 12, 20, respectively.  In this way, the total population in
any support (of $\cS$ or $\cT$) is 60.  We want to consolidate the
information across these maps onto a single map, however, and this
requires to \emph{align} their collections of supports.  To align
collections of supports is to modify the supports of all collections,
in terms of the units they contain, such that the resulting supports
remain contiguous, and the resulting collections are identical.  This
can be viewed as ``swapping'' units between neighboring supports until
the desired alignment is reached.  For example, Figure~\ref{fig:eg-c}
depicts an alignment of collections $\cS$ and $\cT$.  Such an
alignment is obtained from $\cS$ by swapping $u_{1,2}$ and $u_{2,2}$
from $s_4$ (blue) to $s_1$ (green), $u_{2,3}$ from $s_3$ (orange) to
$s_1$ (green), and $u_{2,4}$ from $s_3$ (orange) to $s_4$ (blue).  The
alignment is obtained from $\cT$ by swapping $u_{2,1}$ and $u_{3,1}$
from $t_2$ (yellow) to $t_1$ (green), and $u_{3,2}$ from $t_3$
(orange) to $t_2$ (yellow).

Since any collection of contiguous supports---including supports that
may not be currently present, \eg, a hypothetical $s_5$---is an
alignment, it is desirable to produce an alignment that minimizes the
maximum number of changes in any one collection.  Since $\cS$ and
$\cT$ disagree on 7 units ($u_{2,1}, u_{3,1}, u_{1,2}, u_{2,2},
u_{3,2}, u_{2,3}, u_{2,4}$, annotated with the red dots in
Figure~\ref{fig:eg-c}), one collection must have at least 4 changes
(the other collection having 3 changes), hence the alignment depicted
in Figure~\ref{fig:eg-c} satisfies this criterion.  This need to
adjust leads to a notion of a \emph{distance}, $d(\cS,\cT)$, between a
pair $\cS$ and $\cT$ of collections of spatial supports, namely the
number of swaps needed to transform $\cS$ into $\cT$---this is simply
the number of units on which the pair of collections disagree.  Here,
$d(\cS,\cT) = 7$, and since an alignment is just another collection,
if we denote the alignment of Figure~\ref{fig:eg-c} as collection
$\cA$ of supports, then $d(\cS,\cA) = 4$, and $d(\cT,\cA) = 3$.  Note
that this distance is symmetric.

The units of the different collections being swapped contain
populations, however.  Hence each swap has an associated \emph{cost},
namely the population $p_\cC(u)$ of the unit $u$ in the collection
$\cC$ being swapped.  For example, in $\cS$, swapping $u_{1,2}$ from
$s_4$ (blue) to $s_1$ (green) costs $p_S(u_{1,2}) = 15$.  This idea
leads to a notion of a \emph{weighted} distance $d_w(\cS,\cA)$ between
a pair $\cS$ and $\cA$ of collections of spatial supports, or the
overall cost of the swaps needed to transform $\cS$ into $\cA$.  Here,
$d_w(\cS,\cA) = p_S(u_{1,2}) + p_S(u_{2,2}) + p_S(u_{2,3}) +
p_S(u_{2,4}) = 15 + 15 + 10 + 10 = 50$, while $d_w(\cT,\cA) =
p_T(u_{2,1}) + p_T(u_{3,1}) + p_T(u_{3,2}) = 15 + 15 + 12 = 42$.  Note
that this weighted distance is \emph{not} symmetric, \ie,
$d_w(\cS,\cT) \ne d_w(\cT,\cS)$ in general.  So, more precisely, it is
desirable to produce an alignment $\cA$ that minimizes the maximum
weighted distance $d_w(\cC,\cA)$ between any collection $\cC$ of
spatial supports and $\cA$.  After careful inspection, no alignment
can achieve such a weighted distance less than 50, hence the alignment
depicted in Figure~\ref{fig:eg-c} satisfies this weighted criterion as
well.  We formalize the alignment problem as follows.

\newcommand{\theinput}{
  \textbf{Input}: A base set $U = \{u_1, \dots, u_n\}$ of units over
  some Euclidean geospatial region and set $\sC = \{\cC_1, \dots,
  \cC_k\}$ of collections of spatial supports.  Each unit $u \in U$
  has population $p_\cC(u)$ from $p_\cC : U \rightarrow \mathbb{N}$,
  specific to each collection $\cC \in \sC$.  Each $\cC \in \sC$ is a
  collection $\{s_\cC^1, \dots s_\cC^m\}$ of contiguous supports such
  that: (a) $s \subseteq U$ for each $s \in \cC$; (b) $s \cap t =
  \emptyset$ for any pair $s,t \in \cC$ of distinct supports; and (c)
  $\bigcup_{s \in \cC} s = U$.
}
    
\begin{problem}[The Alignment Problem]~

  \theinput
  
  \textbf{Output}: A collection $\cA$ of contiguous supports which
  satisfies properties (a--c) above, such that $\max\{d_w(\cC,\cA) ~|~
  \cC \in \sC\}$ is minimized.

  \label{prob:alignment}
\end{problem}

Note that properties (a--c) ensure that the set of supports partitions
the base set $U$.  That is, (a) supports contain sets of contiguous
units, (b) pairs of distinct supports are disjoint, and (c) the
supports cover the base set $U$.

In Section~\ref{sec:tractability}, we show that if the Euclidean
geospatial region, that the set $U$ of units is over, is
1-dimensional, then Problem~\ref{prob:alignment}---which we will refer
to as the \emph{Alignment Problem}, when the context is clear---is
solvable in time polynomial in $k$, $m$ and $n$.  In
Section~\ref{sec:hardness}, we show that if the geospatial region is
2-dimensional---which is the typical case in this context of
constructing age-adjusted maps---then the Alignment Problem is
NP-hard, even in the case of 2 collections, each with 2 supports.
Finally, in Section~\ref{sec:heuristic}, we outline a heuristic for
the Alignment Problem in the 2-dimensional case.
Section~\ref{sec:conclusion} concludes the paper and outlines future
work.

\section{Tractability results}
\label{sec:tractability}

In this section, we show that if the Euclidean geospatial region, that
the set $U$ of units is over, is 1-dimensional, then the Alignment
Problem (Problem~\ref{prob:alignment}) is solvable in an amount of
time that is a polynomial function of $k, m$, and $n$.

Consider the set of four collections of three spatial supports over
the common set of 9 spatial units depicted in Figure~\ref{fig:1d}.
Because the 1-dimensional case is so restrictive, each support (\eg,
the orange support) is adjacent to at most two other supports, to the
left (\eg, the green support) and to the right (\eg, the blue
support).  The set of supports in each collection can hence be
enumerated from left to right (from 1 to $m$), and the $i$-th supports
($i \in [1,m]$) of each collection must align.  The disagreements
between $i$-th and $(i+1)$-th supports are then contained in a window
of width at most $n$ (the number $|U|$ of units).  Aligning these two
supports involves scanning this window to find the separator (between
a pair of units) of minimum cost.  For example, in
Figure~\ref{fig:1d}, the disagreements between the first (green) and
second (orange) supports are contained in the transparent window on
the left, with two choices of separator.  Suppose that each unit (of
each collection) of $U$ has a population of 1.  Placing the separator
between $u_2$ and $u_3$ implies coloring $u_3$ of each collection
orange, which has cost 3.  Placing the separator between $u_3$ and
$u_4$ instead, implies coloring $u_3$ of each collection green, which
has cost 1, which is lower.  The disagreements between the second
(orange) and third (blue) supports in Figure~\ref{fig:1d} are
contained in the transparent window on the right, with three choices
of separator.  Of these three choices, placing the separator between
$u_6$ and $u_7$ has cost 2, while the other two choices cost 4 each.

The window between a pair of neighboring supports has width at most
$n$ (the number $|U|$ of units), and height $k$ (the number of
collections).  Each step of a scan within the window from left to
right, for $n+1$ separators, involves updating $k$ units, for at most
$k(n+1)$ operations.  There are $m-1$ such windows (between each pair
of neighboring supports of a set of $m$ supports), hence the overall
number of operations is at most $k(n+1)(m-1) \in O(kmn)$.  Note that
windows may overlap, however, since supports are nonempty and
enumerated from left to right, no window will be contained in another.
This implies that the best separator of the window that begins to the
left of another window will also be to the left of the best separator
of the other window.  If they are at the same position, it simply
means that the corresponding support (say $i$) is empty---the optimal
alignment of $m$ supports comprises $m-1$ supports in this case.

\begin{figure}[!ht]
  \centering
  \includegraphics[width=.7\textwidth]{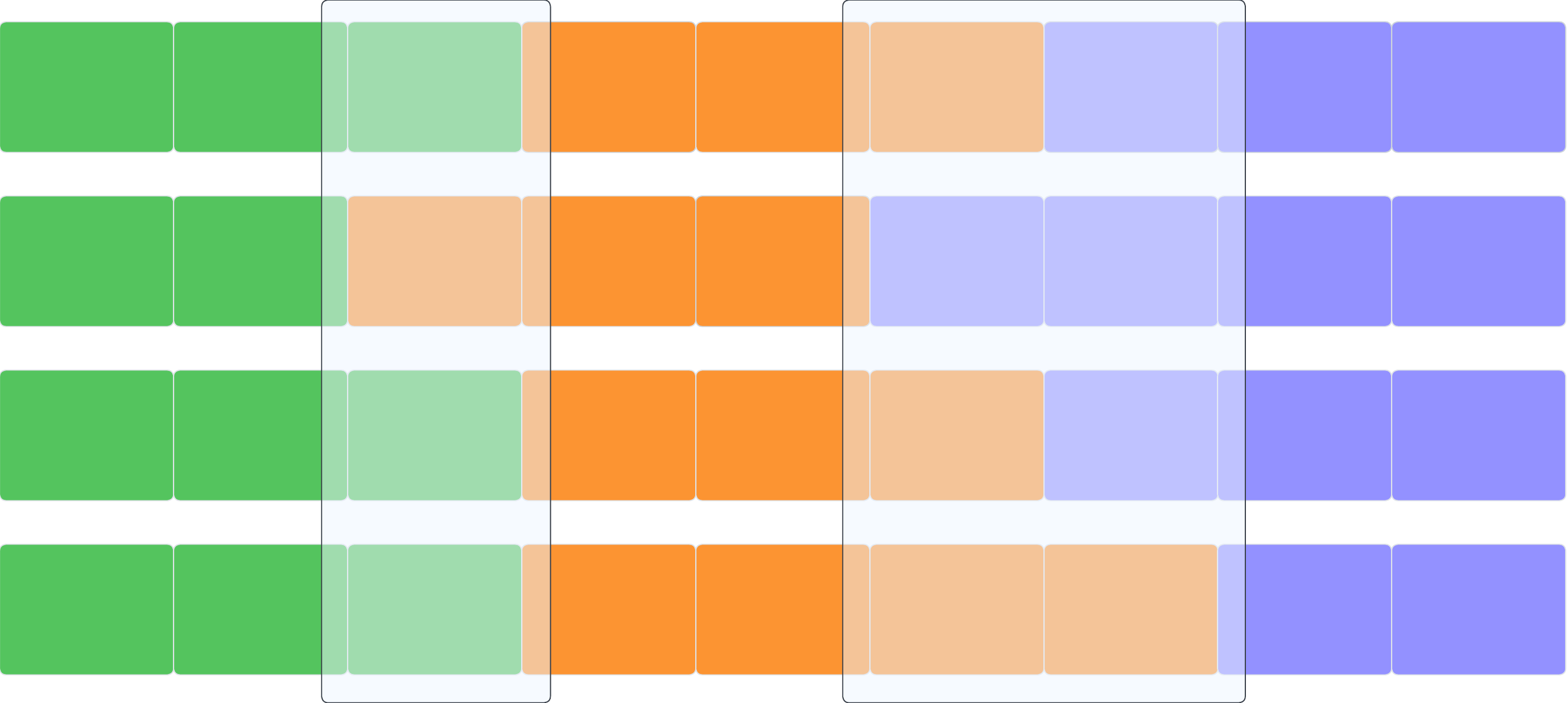}
  \caption{Four collections of three spatial supports (green, orange,
    and blue) over the set $U = \{u_1, u_2, \dots, u_9\}$ of 9 spatial
    units.  All disagreements between (the supports of) any pair of
    collections are contained in the two transparent windows.}
  \label{fig:1d}
\end{figure}

\section{Hardness results}
\label{sec:hardness}

In this section, we show that if the Euclidean geospatial region, that
the set $U$ of units is over, is 2-dimensional, then the Alignment
Problem~\ref{prob:alignment} is NP-hard.  The construction involves 2
collections, each with 2 supports.  This implies that the Alignment
Problem in $d$-dimensions is NP-hard for $k,m,d \ge 2$.

\begin{theorem}
  Problem~\ref{prob:alignment} in $d$-dimensions is NP-hard for $k,m,d
  \ge 2$.
  \label{thm:2d}
\end{theorem}

\begin{proof}
  We first consider the following \emph{decision version} of the
  Alignment Problem~\ref{prob:alignment}.

  \begin{problem}[The Alignment Decision Problem]~

    \theinput

    \textbf{Decision}: Does there exist a collection $\cA$ of
    contiguous supports which satisfy properties (a--c) above, such
    that $\max\{d_w(\cC,\cA) ~|~ \cC \in \sC\} = D$?

    \label{prob:alignment-decision}
  \end{problem}

  Clearly, if the Alignment Decision
  Problem~\ref{prob:alignment-decision} is NP-hard, then so is its
  optimization version, Problem~\ref{prob:alignment}.  We now
  construct a polynomial (Karp) reduction from the following
  Partitioning Problem to the Alignment Decision
  Problem~\ref{prob:alignment-decision}.

  \begin{problem}[The Partitioning Problem]~
  
    \textbf{Input}: A multiset $X = \{x_1, \dots, x_n\}$ of positive
    integers.

    \textbf{Decision}: Does there exist a partition of $X$ into two
    disjoint ($X_1 \cap X_2 = \emptyset$) subsets $X_1 \subseteq X$
    and $X_2 \subseteq X$, such that the difference between the sum
    $\sum_{x \in X_1} x$ of elements in $X_1$ and the sum $\sum_{x \in
      X_2} x$ of elements in $X_2$ is $\Delta$?

    \label{prob:partitioning}
  \end{problem}
    
  The Partitioning Problem~\ref{prob:partitioning} is
  NP-hard~\cite{korf-2009-partitioning}.  We now build a reduction
  from the Partitioning Problem~\ref{prob:partitioning} to the
  Alignment Decision Problem~\ref{prob:alignment-decision} as follows.

  Given an instance $X = \{x_1, \dots, x_n\}$ of the Partitioning
  Problem~\ref{prob:partitioning}, we construct the base set $U \cup
  \{a, b\}$ of spatial units, where $U = \{u_1, \dots, u_n\}$, as
  depicted in Figure~\ref{fig:construction}.  We then introduce the
  two collections $\cS = \{s_1, s_2\}$ and $\cT = \{t_1, t_2\}$ of
  spatial supports, where $s_1 = \{a\} \cup U$, $s_2 = \{b\}$, $t_1 =
  \{a\}$, and $t_2 = U \cup \{b\}$.  For each collection $\cC \in
  \{\cS,\cT\}$, $p_\cC(u_i) = x_i$ for each $i \in \{1, \dots, n\}$,
  and $p_\cC(a) = p_\cC(b) = S+1$, where $S = \sum_{x \in X} x$.  The
  idea is that units $a$ and $b$ have large enough populations that
  they remain in different supports in any alignment of $\cS$ and
  $\cT$ within a given distance threshold.  In this case, the
  alignment is obtained by swapping only the elements of $U$ in either
  $\cS$ or $\cT$, which corresponds to a partition of $U$.  We prove
  the following claim to complete the proof.

  \begin{figure}[!ht]
    \centering
    \includegraphics[width=.7\textwidth]{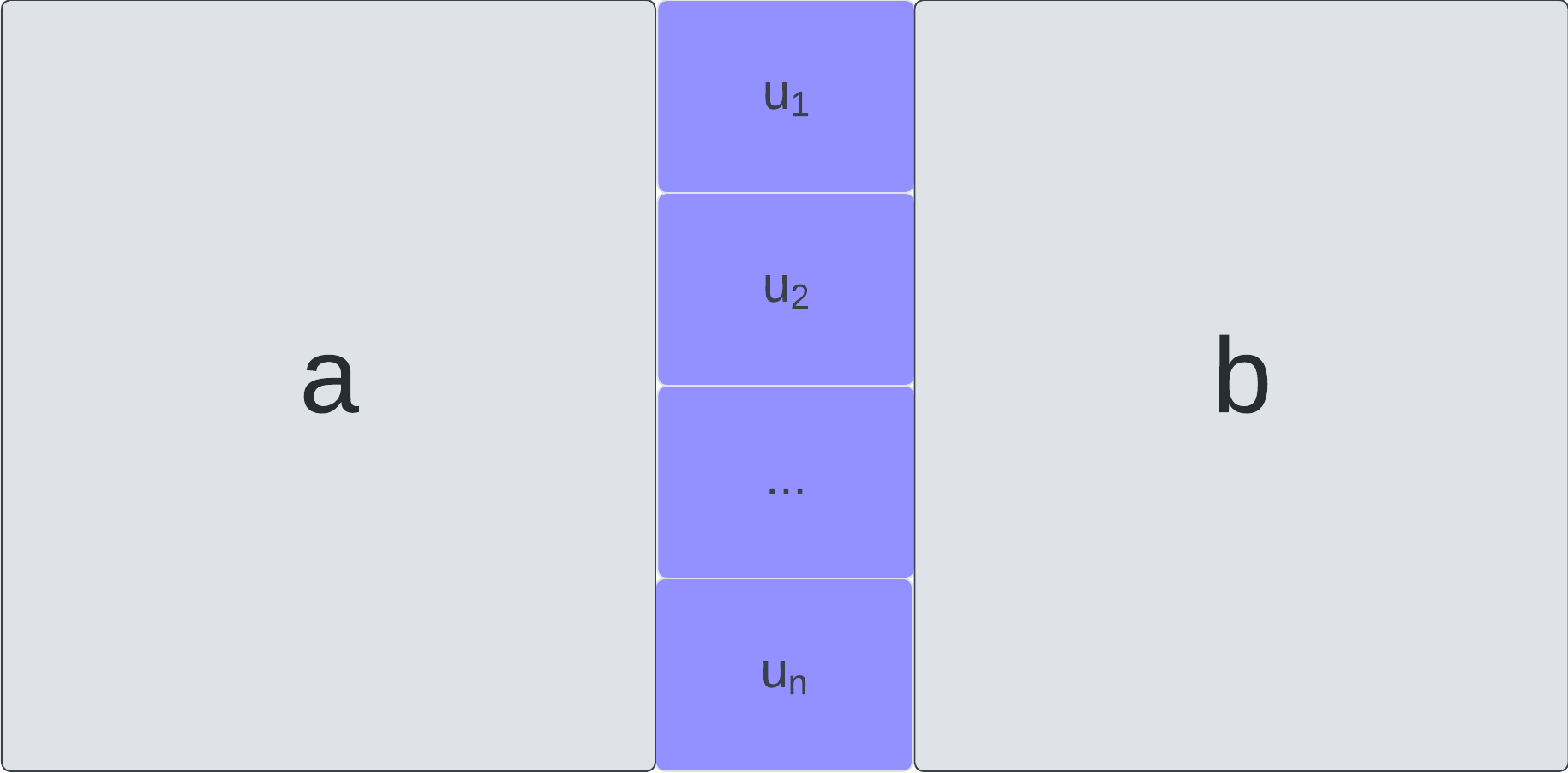}
    \caption{Base set $U \cup \{a, b\} = \{a, u_1, \dots, u_n, b\}$ of
      spatial units.}
    \label{fig:construction}
  \end{figure}

  \paragraph{Claim.}
  There exists a partition of $X$ where the difference between the sum
  of the two parts is $\Delta$ if and only if there exists a
  collection $\cA$ of contiguous supports which satisfies properties
  (a--c) of Problem~\ref{prob:alignment-decision} such that
  $\max\{d_w(\cS,\cA), d_w(\cT,\cA)\} = \frac{S+\Delta}{2}$.

  ($\Rightarrow$) Suppose there exists a partition $(X_1, X_2)$ of $X$
  such that the difference between $S_1 = \sum_{x \in X_1} x$ and $S_2
  = \sum_{x \in X_2} x$ is $\Delta$.  Then consider the collection
  $\cA$ with supports $\{a\} \cup U_1$ and $U_2 \cup \{b\}$---where
  $U_1$ (resp. $U_2$) is the set of units corresponding to $X_1$
  (resp. $X_2$).  By inspecting Figure~\ref{fig:construction}, it is
  clear that $\cA$ satisfies properties (a--c).  It follows that
  $d_w(\cS,\cA) = S_2$, since the units of $U_2$ need to be swapped
  from $s_1$ to $s_2$ in order to transform collection $\cS$ into
  $\cA$.  Conversely, $d_w(\cT,\cA) = S_1$, since the units of $U_1$
  need to be swapped from $t_2$ to $t_1$ in order to transform
  collection $\cT$ into $\cA$.  Suppose, without loss of generality,
  that $S_1$ is the larger sum, \ie, $S_1 > S_2$, hence $S_1 - S_2 =
  \Delta$.  Since $S_1 + S_2 = S$, it follows that $S_1 - (S - S_1) =
  \Delta$, then $2S_1 = S + \Delta$, and $S_1 = \frac{S+\Delta}{2}$.
  Since $S_1$ is the larger sum, it follows that $\max\{d_w(\cS,\cA),
  d_w(\cT,\cA)\} = S_1 = \frac{S+\Delta}{2}$.

  ($\Leftarrow$) Suppose there exists a collection $\cA$ of contiguous
  supports which satisfies properties (a--c) of
  Problem~\ref{prob:alignment-decision} such that $\max\{d_w(\cS,\cA),
  d_w(\cT,\cA)\} = \frac{S+\Delta}{2}$.  Since $\Delta \le S$, it
  follows that $\frac{S+\Delta}{2} \le S$, hence both $d_w(\cS,\cA)
  \le S$ and $d_w(\cT,\cA) \le S$.  Therefore, it must be the case
  that units $a$ and $b$ are in two different supports of $\cA$,
  otherwise a swap of weight at least $S+1$ would be needed to
  transform $\cS$ or $\cT$ into $\cA$, contradicting the assumption
  that $d_w(\cS,\cA) \le S$ and $d_w(\cT,\cA) \le S$.  Suppose,
  without loss of generality, that $d_w(\cS,\cA) =
  \frac{S+\Delta}{2}$.  Then the support of $\cA$ that contains unit
  $a$ must agree with $s_1$ on a subset $U_1 \subseteq U$ of units
  with $\sum_{u \in U_1} p_\cS(u) = \frac{S-\Delta}{2}$.  Since
  $\sum_{u \in U} p_\cS(u) = S$, it follows that $U_2 = U \setminus
  U_1$ has $\sum_{u \in U_2} p_\cS(u) = S - \frac{S-\Delta}{2} =
  \frac{S+\Delta}{2}$.  Since $\frac{S+\Delta}{2} - \frac{S-\Delta}{2}
  = \Delta$, it follows that $(U_1, U_2)$ is a partition of $U$ such
  that the difference between the sum of the (populations of the) two
  parts is $\Delta$, hence there exists such a partition of $X$.
\end{proof}

\section{A heuristic for the 2-dimensional case}
\label{sec:heuristic}

In this section, we give a (polynomial time) heuristic for the
Alignment Problem~\ref{prob:alignment} when the geospatial region is
2-dimensional.  We first give a heuristic for a pair of collections in
2 dimensions (in Sec~\ref{sec:pair}), and then show how this can be
extended to a general set of collections in 2 dimensions (in
Sec~\ref{sec:general}), \ie, to the general Alignment Problem in 2
dimensions.

\subsection{Aligning a pair of collections in 2 dimensions}
\label{sec:pair}

Figure~\ref{fig:eg} depicts an example of this special case of
aligning only two collections of supports---(a) and (b) in this case.
Since there are only two collections, the idea is that each support in
one collection is matched up with another support in the other
collection, then these pairs of supports are aligned with each other.
For example, in the instance depicted in Figure~\ref{fig:eg}, $s_i \in
\cS$ is matched up with $t_i \in \cT$ for each $i \in \{1,2,3,4\}$.
We first need the following definition of shared units graph.

\begin{definition}[Shared units Graph $G_x$]~

  The shared units graph, for pair $\cS, \cT$ of collections of
  spatial supports, is the weighted graph $G_x = (\cS, \cT, E = \cS
  \times \cT, w : E \rightarrow \mathbb{R})$, where each edge $e \in
  E$ has weight $w(e) = \sum_{u \in s \cap t} p_\cS(u) + p_\cT(u)$.
  \label{def:gx}
\end{definition}

The shared units graph $G_x$ gives a measure of the weighted overlap
of the supports between each collection.  This information will be
used to match up the supports between each collection.  More
precisely, supports between each collection will be paired up
according to a maximum weight perfect matching in $G_x$.  For example,
Figure~\ref{fig:gx} depicts the shared units graph $G_x$ of the
instance depicted in Figure~\ref{fig:eg}.  Here, \eg, $w(s_1,t_1) =
35$ because $s_1 \cap t_1 = u_{1,1}$, and $p_\cS(u_{1,1}) +
p_\cT(u_{1,1}) = 20 + 15 = 35$.  After careful inspection, the maximum
weight perfect matching of the graph depicted in Figure~\ref{fig:gx}
is $\{(s_1, t_1), (s_2, t_2), (s_3, t_3), (s_4, t_4)\}$, of weight 35
+ 70 + 88 + 70 = 263.  This is why the supports of this instance of
Figure~\ref{fig:eg} are paired up accordingly, as represented by the
matching colors.

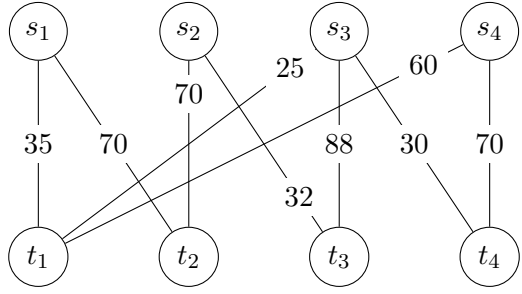
\begin{figure}[!ht]
  \centering
  \begin{tikzpicture}[node distance=2cm]
    \node[n] (s1) {$s_1$};
    \node[n] (s2) [right of=s1] {$s_2$};
    \node[n] (s3) [right of=s2] {$s_3$};
    \node[n] (s4) [right of=s3] {$s_4$};
    \node[n] (t1) [below of=s1, node distance=3cm] {$t_1$};
    \node[n] (t2) [right of=t1] {$t_2$};
    \node[n] (t3) [right of=t2] {$t_3$};
    \node[n] (t4) [right of=t3] {$t_4$};
    \path (s1) edge node[e] {$35$} (t1);
    \path (s1) edge node[e] {$70$} (t2);
    \path (s2) edge node[e,pos=.2] {$70$} (t2);
    \path (s2) edge node[e,pos=.8] {$32$} (t3);
    \path (s3) edge node[e,pos=.1] {$25$} (t1);
    \path (s3) edge node[e] {$88$} (t3);
    \path (s3) edge node[e] {$30$} (t4);
    \path (s4) edge node[e,pos=.1] {$60$} (t1);
    \path (s4) edge node[e] {$70$} (t4);
  \end{tikzpicture}
  \caption{The shared units graph $G_x$ (Definition~\ref{def:gx}) of
    the instance $\cS, \cT$ depicted in Figure~\ref{fig:eg}, where the
    populations $p_\cS(u)$ of each unit $u$ in $s_1, s_2, s_3, s_4$ of
    $\cS$ are 20, 20, 10, 15, respectively, while the populations
    $p_\cT(u)$ of each unit $u$ in $t_1, t_2, t_3, t_4$ of $\cT$ are
    15, 15, 12, 20, respectively.  Edges of zero weight are not shown
    for easier readability.}
  \label{fig:gx}
\end{figure} 

In general, a maximum weight perfect matching $M$ in a weighted graph
$G = (V,E,w)$ can be found in time $O(|V| \log|V| + |V| \cdot |E|)$
using a Fibonacci heap~\cite{fredman-1987-fib}.  Note that, in
general, the number of supports of one collection may be different
than the other.  Suppose, without loss of generality, that $|\cS| >
|\cT|$ for some pair $\cS, \cT$ of collections of supports.  In this
case, a perfect matching $M$ is found in the shared units graph $G_x$
for $\cS, \cT$, and each remaining unmatched support in $\cS$ is
associated to one of the supports in $\cT$.  After this process, each
support in $t \in \cT$ will be matched up with a support $s \in \cS$,
and possibly another subset $S \subseteq \cS$ of supports.  The idea
is that $\{s\} \cup S$ should be a contiguous set of supports in
$\cS$.  Hence, the criteria for associating each unmatched support in
$\cS$ with some support $t \in \cT$ is that the resulting subset $S
\subseteq \cS$ associated with $t$ is such that $\{s\} \cup S$, while
$w(s',t)$ for each $s' \in \cS$ is maximized.  Since we would expect
$|\cS| - |\cT|$ to typically be a constant, all combinations could be
tried to achieve this.  Hence the overall procedure of pairing each
support (or set of supports) of $\cS$ with a support in $\cT$ takes
polynomial time.  In cases where this does not hold ($|\cS| - |\cT|$
is not a constant), a more efficient algorithm for determining (or
approximating) this is the subject of future work.

The purpose of pairing each support $s$ (or set $S$ of supports) in
one collection $\cS$ to another support $t$ in the other collection
$\cT$ is to determine how the trading procedure, for aligning pair
$\cS,\cT$ of collections, operates.  In particular, each support $s
\in \cS$ (resp., set $S \subseteq \cS$) and its counterpart $t \in
\cT$ swap units with their respective neighbors until they are aligned
(are on the same set of units).  Figure~\ref{fig:eg-c} depicts an
alignment of collections $\cS$ (of Figure~\ref{fig:eg-a}) and $\cT$
(of Figure~\ref{fig:eg-b}) according to the maximum weight matching
$\{(s_1, t_1), (s_2, t_2), (s_3, t_3), (s_4, t_4)\}$ in the shared
units graph $G_x$ of $\cS, \cT$ depicted in Figure~\ref{fig:gx}.
While Figure~\ref{fig:eg-c} depicts an optimal alignment of these
collections $\cS$ and $\cT$, we outline a polynomial-time heuristic
for the general case, since it is NP-hard (see
Section~\ref{sec:hardness}).

Aligning a pair of collections of supports in 2 dimensions is a
partitioning problem (the NP-hardness proof of this case based on a
reduction from the Partitioning Problem~\ref{prob:partitioning}).
Hence, we apply a straightforward greedy partitioning heuristic to the
problem which is slightly more general than the
longest-processing-time-first (LPT) scheduling
heuristic~\cite{graham-1969-timing,coffman-1976-lpt}.  In LPT
scheduling, we are given a set of numbers and a positive integer $m$,
and the goal is to partition this set into $m$ subsets such that the
largest sum of any subset (in terms of the values of its elements) is
minimized.  This problem is NP-hard, because its decision version (the
Partitioning Problem~\ref{prob:partitioning}) is NP-hard.  The LPT
scheduling heuristic is to order the elements of the set from largest
to smallest, and iteratively place each element from this sorted list
in the subset (of $m$ subsets) with the smallest sum so far, until all
elements are placed.  In our variation, we have two sets, one for each
of the pair of collections.  That is, given the set $U' \subseteq U$
of units on which the pair, say $\cS, \cT$, of collections disagree
(based on the pairing of supports between $\cS$ and $\cT$), we first
sort $U'$ in descending order of population according to both $p_\cS$
and $p_\cT$, separately.  We then partition $U'$ into two parts $S$
and $T$, representing $\cS$ and $\cT$, respectively.  This is an
iterative process which considers the part with the currently lower
population (breaking ties arbitrarily), and adds the next element to
this part according to its ordering.  For example, if part $S$ has the
currently lower population, then $\sum_{u \in S} p_\cS(u) < \sum_{u
  \in T} p_\cT(u)$, and we would add the next largest element of $U'$
(according to $\cS$) to $S$.  The iteration terminates when all
elements of $U'$ have been assigned to either $S$ or $T$.

For example, consider the instance $\cS, \cT$ depicted in
Figure~\ref{fig:eg}, where the populations $p_\cS(u)$ of each unit $u$
in $s_1, s_2, s_3, s_4$ of $\cS$ are 20, 20, 10, 15, respectively,
while the populations $p_\cT(u)$ of each unit $u$ in $t_1, t_2, t_3,
t_4$ of $\cT$ are 15, 15, 12, 20, respectively.  Here, the set $U'$ of
units on which $\cS, \cT$ disagree is $U' = \{u_{2,1}, u_{3,1},
u_{1,2}, u_{2,2}, u_{3,2}, u_{2,3}, u_{2,4}\}$, annotated with the red
dots in Figure~\ref{fig:eg-c}.  Note that $U'$ is currently ordered in
reverse lexicographic order, starting from the lower left corner
($u_{1,1}$) and moving to the right, row by row, upward.  By sorting
this order in a stable way (the order of identical elements is not
disturbed) according to $p_\cS$, it becomes $u_{2,1} (20), u_{3,1}
(20), u_{3,2} (20), u_{1,2} (15), u_{2,2} (15), u_{2,3} (10), u_{2,4}
(10)$.  By sorting this order in a stable way according to $p_\cT$, it
becomes $u_{2,4} (20), u_{2,1} (15), u_{3,1} (15), u_{1,2} (15),
u_{2,2} (15), u_{2,3} (15), u_{3,2} (12)$.  The iteration then takes
the steps indicated by Table~\ref{tab:steps}, starting with empty
parts $S$ and $T$.  After this process completes, the resulting parts
$S$ and $T$ then join (take their current color in) the corresponding
supports $\cS$ and $\cT$, respectively, in order to produce the
alignment.  For example, the partitioning outlined in
Table~\ref{tab:steps} produces the alignment depicted in
Figure~\ref{fig:eg-c}, which is optimal.

\begin{table}[!ht]
  \centering
  \begin{tabular}{|l|l|l|l|}
    \hline
    step & action & part $S$ & part $T$ \\
    \hline
    0 & initialize $S$ \& $T$ & $\emptyset$ & $\emptyset$ \\
    1 & $S \leftarrow u_{2,1} (20)$ & $\{u_{2,1}\}$ (20) & $\emptyset$ \\
    2 & $T \leftarrow u_{2,4} (20)$ & $\{u_{2,1}\}$ (20) & $\{u_{2,4}\}$ (20) \\
    3 & $S \leftarrow u_{3,1} (20)$ & $\{u_{2,1}, u_{3,1}\}$ (40) & $\{u_{2,4}\}$ (20) \\
    4 & $T \leftarrow u_{1,2} (15)$ & $\{u_{2,1}, u_{3,1}\}$ (40) & $\{u_{2,4}, u_{1,2}\}$ (35) \\
    5 & $T \leftarrow u_{2,2} (15)$ & $\{u_{2,1}, u_{3,1}\}$ (40) & $\{u_{2,4}, u_{1,2}, u_{2,2}\}$ (50) \\
    6 & $S \leftarrow u_{3,2} (20)$ & $\{u_{2,1}, u_{3,1}, u_{3,2}\}$ (60) & $\{u_{2,4}, u_{1,2}, u_{2,2}\}$ (50) \\
    7 & $T \leftarrow u_{2,3} (15)$ & $\{u_{2,1}, u_{3,1}, u_{3,2}\}$ (60) & $\{u_{2,4}, u_{1,2}, u_{2,2}, u_{2,3}\}$ (65) \\
    \hline
  \end{tabular}
  \caption{Steps taken by the greedy approach to create parts $S$ and
    $T$.}
  \label{tab:steps}
\end{table}

In general, our greedy approach does not produce an optimal solution
to the Alignment Problem~\ref{prob:alignment}, however an upper bound
on the quality of the solution, $\max\{d_w(\cS,\cA), d_w(\cT,\cA)\}$,
can be obtained based on known approximation factors for LPT
scheduling~\cite{graham-1969-timing,coffman-1976-lpt}.  Given some
instance $\cS,\cT$ to the Alignment Problem, let $\sigma(u) =
\{p_\cS(u), p_\cT(u)\}$.  For example, from the instance mentioned
above, $\sigma(u_{2,1}) = \{p_\cS(u_{2,1}), p_\cT(u_{2,1})\} = \{20,
15\}$.  For some set $U'$ of units, let $M(U') = \{\max(\sigma(u)) ~|~
u \in U'\}$, the maximum values of the pairs $\sigma(u)$ of
populations represented by each $u$.  Note that our greedy approach
obtains a partitioning by effectively applying LPT scheduling to
$M(U')$, where $U'$ is the set of units on which a pair $\cS,\cT$ of
collections disagree (see Table~\ref{tab:steps}), hence we can bound
its quality based on known bounds for LPT scheduling.  It is known
that applying LPT scheduling to a set guarantees a solution that is
within a factor of $\frac{4m-1}{3m}$ times the optimal (minimum)
largest sum of any of the $m$
subsets~\cite{graham-1969-timing,coffman-1976-lpt}.  Supposing we
partition $M(U')$ into a pair ($m=2$) of parts using LPT scheduling,
let $A$ be the part with the larger sum, and $A^*$ be the part with
the larger sum in the optimal partitioning of $M(U')$ into two parts.
It then follows that \[ A \le \frac{4(2)-1}{3(2)} = \frac{7}{6} \cdot
A^*. \] Because the elements that we are partitioning are indivisible,
we know that \[ A^* \le \frac{\sum(M(U'))}{2} + \max(M(U')), \] where
$\sum(X) = \sum_{x \in X} x$, a short form for the sum of all values
in a set $X$ of values.  It follows that
\begin{equation}
  A \le \frac{7}{6} \left[ \frac{\sum(M(U'))}{2} + \max(M(U'))
    \right].
  \label{eq:auxbound}
\end{equation}
Let part $S$ be the set of units represented by part $A$.  The units
of $S$ were chosen based on the largest values from $M(U')$ at the
time, as represented by $A$.  The units of $S$ are used to transform
collection $\cS$ of supports into another collection $\cA$ of supports
(while the remaining units of $U' \setminus S$ are used to transform
collection $\cT$ into $\cA$).  It follows that $d_w(\cS,\cA) = m(S)$,
where $m(S) = \{\min(\sigma(u)) ~|~ u \in S\}$, the minimum values of
the pairs $\sigma(u)$ of populations represented by each $u$.  Since
$m(S) \le A$, by design, and $A$ is the larger part, \ie, $m(U'
\setminus S) \le A$ as well, it follows from
Equation~\ref{eq:auxbound} that
\begin{equation}
  \max\{d_w(\cS,\cA), d_w(\cT,\cA)\} \le \frac{7}{6} \left[
    \frac{\sum(M(U'))}{2} + \max(M(U')) \right].
  \label{eq:bound}
\end{equation}
Since a typical instance $\cS,\cT$ will contain many units $u$ which
do not differ much in $p_\cS(u)$ and $p_\cT(u)$, nor is $\max(M(U')$
and $\min(m(U')$ expected to differ by much, each collection $\cS$ and
$\cT$ will typically contribute close to half of their weight to the
alignment $\cA$.

There remain some small and final details to address in this
heuristic.  One detail is that the units of $U'$, on which $\cS$ and
$\cT$ disagree, cannot be placed into parts arbitrarily.  Rather, the
parts must be such that swapping their units results in an alignment
$\cA$ whose supports are contiguous (see the Alignment
Problem~\ref{prob:alignment}).  The example outlined in
Table~\ref{tab:steps} happens to create a contiguous set of supports,
as depicted in Figure~\ref{fig:eg-c}.  However, if units $u_{1,2}$ and
$u_{2,2}$ were assigned to part $S$ instead of part $T$, the green
support would not be contiguous, for example.  In a general instance,
such a constraint only needs to be minded for each contiguous set $U'$
of units on which $\cS$ and $\cT$ disagree.  For each such contiguous
set, some small local shuffles could be applied to each ordering of
$U'$ according to $p_\cS$ and $p_\cT$, respectively.  Another solution
could be to apply the iteration to $S$ and $T$ as is, but skipping any
greedy choice which violates contiguity.  In any case, the iteration
will be no worse than (unordered) list
scheduling~\cite{graham-1969-timing}.  In this case, it is known that
applying list ordering to a set guarantees a solution within a factor
of $2-\frac{1}{m}$ times the optimal (minimum) largest sum of any of
the $m$ subsets.  Since $m=2$ in this case, it follows that this
factor is $\frac{3}{2}$, and the same analysis as above can be
applied.  Since there will be few such constraints in the typical
instance, and they only apply to contiguous sets of units of $U'$,
which will be typically small, the solution is expected to be much
closer to $\frac{7}{6}$ (see Equation~\ref{eq:bound}) than
$\frac{3}{2}$, in practice.  The other detail is the unmatched
supports (in, \eg, $\cS$) associated with some support (\eg, $t \in
\cT$).  In this case, for the support in the final alignment $\cA$
that represents these will present another alignment subproblem within
that support, where this support could be split into several parts.
Since such supports are expected to be small in general, all ways to
align this support could be tried.  Nonetheless, a more systematic
procedure for minding such constraints, along with a more definite
approximation factor is the subject of future work.

\subsection{The Alignment Problem in 2 dimensions}
\label{sec:general}

We now outline how to extend the techniques used in the heuristic of
Section~\ref{sec:pair} to a general set $\sC = \{\cC_1, \dots,
\cC_k\}$ of collections in 2 dimensions, \ie, to the general Alignment
Problem in 2 dimensions.  We first need to match supports across all
collections $\sC$ in order to align them.  This amounts to finding a
maximum weight perfect matching in a complete $k$-uniform hypergraph
across all ($k$) collections of supports.  We need the following
definition of a shared units hypergraph, analogous to the shared units
graph of Definition~\ref{def:gx}.

\begin{definition}[Shared units Hypergraph $H_x$]~

  The shared units hypergraph, for a set $\sC = \{\cC_1, \dots,
  \cC_k\}$ of collections of spatial supports, is the weighted
  hypergraph $H_x = (\cC_1, \dots, \cC_k, E = \cC_1 \times \cdots
  \times \cC_k, w(e) : E \rightarrow \mathbb{R})$, where each
  hyperedge $e \in E$ has weight $w(e) = \sum_{(s,t) \in e^2} \sum_{u
    \in s \cap t} p_{\cC(s)}(u) + p_{\cC(t)}(u)$, where $\cC(s)$ is
  the collection that support $s$ belongs to.
  \label{def:hx}
\end{definition}

The weight of a hyperedge $e$ of $H_x$ is effectively the weighted
overlap of the set of $k$ supports, one from each collection $\cC_1,
\dots, \cC_k$, represented by $e$, in terms of the weighted overlap
between each pair $s,t$ of supports from $e$.  Finding a perfect
matching in $H_x$ is
NP-hard~\cite{karp-1972-reducibility,garey-johnson}.  This problem is
a special case of the $k$-set packing problem, which can be
approximated within a factor of $\frac{k+1+\varepsilon}{3}$ times the
optimal packing~\cite{cygan-2013-3DM,furer-2014-set}.  Since $k$ is
typically a small constant (less than 10, for example), this bound is
acceptable in practice.  When the number of supports in the
collections differ, a perfect matching $M$ (of size $\argmin_{\cC \in
  \sC} |\cC|$) is found in the shared units hypergraph $H_x$, and each
remaining unmatched support $s$ in any collection $\cC \in \sC$ is
associated to one of the hyperedges in $M$ of maximum overlap with
$s$.  Similarly to the case with a pair of collections, the hyperedge
that $s$ joins should maintain a contiguous set of supports in
$\cC(s)$, the collection that support $s$ belongs to.  Since, again,
$\argmax_{\cC \in \sC} |\cC| - \argmin_{\cC \in \sC} |\cC|$ should
typically be a constant, all combinations of hyperedges for $s$ to
join could be tried to achieve contiguity, however a more efficient
algorithm for determining these choices is the subject of future work.

Analogously to the case with a pair of collections (of
Section~\ref{sec:pair}), matching up sets of supports across
collections is to determine how the trading procedure, for aligning
collections $\sC$, operates.  In particular, each set of supports from
matching $M$ in $H_x$ (with the extra unmatched supports joined later)
swap units with their respective neighbors until they are aligned.
Similar to the case with pairs, the matching $M$ gives rise to a set
$U' \subseteq U$ of units on which some pair $\cC,\cC' \in \sC$ of
collections disagree.  Each such unit must be assigned to some support
(in $M$) in a way that minimizes overall cost.  Aligning the units of
$U'$ in this way is again a type of partitioning problem, which could
also be approximated using LPT scheduling, however a slightly more
general partitioning problem is more appropriate in this case.  In
particular, this is more closely related to a case of the problem of
fair item allocation~\cite{demko-1988-equitable}, with additive
preferences~\cite{bouveret-2010-fair} and positively valued goods.
Note that there exist versions with negatively weighted goods, or
chores, as well~\cite{aziz-2017-chores}.

The input to this problem is a set $N$ of $|N| = n$ agents and a set
$M$ of $|M| = m$ items.  We use the elements $i \in N$ of a set $N$
and its corresponding indices $i \in \{1,\dots,n\}$ interchangeably,
when the context is clear.  Each agent $i \in N$ attaches a value
$v_i(j)$ to item $j \in M$, where $v_i(j) \in \mathbb{Z}^+ ~\forall i
\in N ~\forall j \in M$.  We also overload the meaning of $v$ for
subsets $S \subseteq M$, where $v_i(S) = \sum_{j \in S} v_i(j)$, since
values are additive.  Let $\Pi_n(M)$ be the collection of all
partitionings of set $M$ into $n$ parts.  The goal is to find a
partitioning, in $\Pi_n(M)$, that gives each agent their fairest share
of value from the items.  A common formalization for this is the
\emph{maximin share}~\cite{barman-2020-maximin} of agent $i \in N$
from a set $M$ of items, which is
\begin{equation}
  \mu_i^n(M) = \max_{(M_1, M_2, \dots, M_n) \in \Pi_n(M)} ~ \min_{k
    \in N} ~ v_i(M_k).
  \label{eq:maximin}
\end{equation}
The idea is that if agent $i \in N$ were to divide items $M$ into $n$
parts, and then other agents chose how these $n$ parts were
distributed among the $n$ agents, then agent $i$ would partition the
items such that value $v_i$ of the smallest part $M_k$ is maximized.
In fair item allocation, the goal is to partition the items such that
each agent $i \in N$ has a value that is closest to their maximin
share $\mu_i$ as possible.  An important approximation result is that
a partitioning $(M_1, \dots, M_n) \in \Pi_n(M)$ which satisfies
\begin{equation}
  v_i(M_i) \ge \frac{2}{3} \mu_i^n(M) ~\forall i \in N
  \label{eq:approx}
\end{equation}
can be found in polynomial time~\cite{barman-2020-maximin}.

Our problem of aligning each unit of $U'$ is closely related to this
problem, in that each collection $\cC \in \sC$ is an agent, and each
unit $u \in U'$ is an item that gets assigned to some collection when
aligned, where $v_\cC(u) = p_\cC(u)$.  The only difference is that the
collection $\cC$ to which $u$ is assigned avoids the cost $p_\cC(u)$,
while every other collection $\cC' \in \sC \setminus \{\cC\}$ incurs
(at most) its corresponding cost $p_{\cC'}(u)$.  Since we want to
minimize the maximum cost to any collection (see
Problem~\ref{prob:alignment}), we are rather aiming, for each
collection $i \in N = \sC$, given set $M = U'$ of units, to minimize
\begin{equation}
  \gamma_i^n(M) = \min_{(M_1, M_2, \dots, M_n) \in \Pi_n(M)} ~ \max_{k
    \in N} \sum_{j \in N \setminus \{k\}} v_i(M_j).
  \label{eq:minimax}
\end{equation}
Note that this is equivalent to \[ \gamma_i^n(M) = \min_{(M_1, M_2,
  \dots, M_n) \in \Pi_n(M)} ~ \max_{k \in N} ~ C_i(M) - v_i(M_k), \]
where $C_i(M) = \sum_{j \in M} v_i(j)$.  Since $C_i(M)$ does not
depend on the partition chosen from $\Pi_n(M)$, it follows that \[
\gamma_i^n(M) = C_i(M) + \min_{(M_1, M_2, \dots, M_n) \in \Pi_n(M)} ~
\max_{k \in N} ~ -v_i(M_k). \] In pulling the minus sign through to
the front, it follows that \[ \gamma_i^n(M) = C_i(M) - \max_{(M_1,
  M_2, \dots, M_n) \in \Pi_n(M)} ~ \min_{k \in N} ~ v_i(M_k). \] We
can then substitute the lefthand side of Equation~\ref{eq:maximin}
with the righthand side to obtain
\begin{equation}
  \gamma_i^n(M) = C_i(M) - \mu_i^n(M).
  \label{eq:equivalence}
\end{equation}
Then, based on the result of Equation~\ref{eq:approx}, it follows that
a partitioning $(M_1, \dots, M_n) \in \Pi_n(M)$ which satisfies
\begin{equation}
  \sum_{j \in N \setminus \{i\}} v_i(M_j) \le C_i(M) - \frac{2}{3}
  \mu_i^n(M) ~\forall i \in N
  \label{eq:maxbound}
\end{equation}
can be found in polynomial time.

Placing this result in the notation of our problem, where $N = \sC$,
and $M = U'$, it follows that $\sum_{j \in N \setminus \{i\}} v_i(M_j)
= \sum_{\cC' \in \sC \setminus \{\cC\}} p_\cC({U'}_{\cC'})$, where
${U'}_{\cC'}$ are the units from $U'$ assigned to collection $\cC'$ in
the alignment $\cA$ represented by partitioning $(U_1, \dots, U_n) \in
\Pi_n(U')$, and the meaning of $p_\cC$ has been overloaded for sets,
where $p_\cC(U') = \sum_{u \in U'} p_\cC(u)$.  Observe that
$d_w(\cC,\cA) = \sum_{\cC' \in \sC \setminus \{\cC\}}
p_\cC({U'}_{\cC'})$.  Then it follows from Equation~\ref{eq:maxbound}
that
\begin{equation}
  d_w(\cC,\cA) \le \sum_{u \in U'} p_\cC(u) - \frac{2}{3}
  \mu_\cC^n(U') ~\forall \cC \in \sC,
  \label{eq:distbound}
\end{equation}
where $U'$ is the set of units on which some pair of collections of
$\sC$ disagree, and $\mu_\cC^n(U')$ is the maximin share of collection
$\cC$ from set $U'$ of units, where the value of a unit is $p_\cC(u)$.
This guarantees a bound on $\max\{d_w(\cC,\cA) ~|~ \cC \in \sC\}$ (see
Problem~\ref{prob:alignment}) which can be obtained in polynomial
time.  The approximation result in of Equation~\ref{eq:approx}
from~\cite{barman-2020-maximin} relies on a complex preprocessing step
from~\cite{bouveret-2016-conflicts} in order to guarantee this
theoretical bound.  However, in practice will plan to use a more
straightforward approach based on the envy-graph
procedure~\cite{lipton-2004-envy,barman-2020-maximin,aziz-2017-chores},
which is the common approach used for fair item allocation.  Such an
approach iterates through the items, assigning them to agents.  If
ever an envy cycle arises in this process---a directed cycle on a set
of agents where each agent values more the intermediate set of items
of her neighbor---then this cycle is broken by shifting this cycle one
step in opposite direction.  This process continues until all items
are assigned to some agent.  While there are many theoretical results
in this area of fair item allocation, there exist some practical
results such as \url{spliddit.org}~\cite{goldman-2014-spliddit}, based
on theoretical results in~\cite{procaccia-2014-fair}.  We plan to use
or follow these ideas in devising a practical algorithm for our
problem.  Similarly to the case of pairs in two dimensions (of
Section~\ref{sec:pair}), maintaining contiguity, and how to manage the
unmatched supports associated after the matching was computed.  An
efficient implementation addressing all of these details is the
subject of future work.

\section{Conclusion}
\label{sec:conclusion}

In this paper, we introduce an alignment problem for reconciling
misaligned boundaries of regions comprising spatial units.  While the
general problem is combinatorially (NP-) hard, and a rather trivial
case (in 1 dimension) is tractable, we devise some heuristics for the
case in 2-dimensions, since it has applications in many geospatial
problems such as resource allocation or building disease maps.

Future work entails further investigation into small details such as
maintaining contiguity in the trading procedure, and how to
systematically handle collections with different numbers of supports.
Developing an implementation which works efficiently in practice is
also the subject of future work, so that can be applied to real
geospatial problems such as resource allocation and automated map
building.

\section*{Acknowledgements}

The authors would like to thank Alexander Zelikovsky for some helpful
discussions on interpreting the approximation results.


\begin{thebibliography}{10}

\bibitem{aziz-2017-chores}
Haris Aziz, Gerhard Rauchecker, Guido Schryen, and Toby Walsh.
\newblock Algorithms for max-min share fair allocation of indivisible chores.
\newblock In {\em Thirty-First AAAI Conference on Artificial Intelligence},
  volume~31, 2017.

\bibitem{barman-2020-maximin}
Siddharth Barman and Sanath~Kumar Krishnamurthy.
\newblock Approximation algorithms for maximin fair division.
\newblock {\em ACM Trans. Econ. Comput.}, 8(1), 2020.

\bibitem{bouveret-2010-fair}
Sylvain Bouveret, Ulle Endriss, and J\'{e}r\^{o}me Lang.
\newblock Fair division under ordinal preferences: Computing envy-free
  allocations of indivisible goods.
\newblock In {\em Proceedings of the 2010 Conference on ECAI 2010: 19th
  European Conference on Artificial Intelligence}, page 387–392, 2010.

\bibitem{bouveret-2016-conflicts}
Sylvain Bouveret and Michel Lemaître.
\newblock Characterizing conflicts in fair division of indivisible goods using
  a scale of criteria.
\newblock In {\em Autonomous Agents and Multi-Agent Systems}, volume~30, pages
  259--29, 2016.

\bibitem{coffman-1976-lpt}
E.~G. Coffman and Ravi Sethi.
\newblock A generalized bound on {LPT} sequencing.
\newblock In {\em Proceedings of the 1976 ACM SIGMETRICS Conference on Computer
  Performance Modeling Measurement and Evaluation}, page 306–310. Association
  for Computing Machinery, 1976.

\bibitem{cygan-2013-3DM}
Marek Cygan.
\newblock Improved approximation for 3-dimensional matching via bounded
  pathwidth local search.
\newblock In {\em 2013 IEEE 54th Annual Symposium on Foundations of Computer
  Science}, pages 509--518, 2013.

\bibitem{demko-1988-equitable}
Stephen Demko and Theodore~P. Hill.
\newblock Equitable distribution of indivisible objects.
\newblock {\em Mathematical Social Sciences}, 16(2):145--158, 1988.

\bibitem{fredman-1987-fib}
M.L. Fredman and R.E. Tarjan.
\newblock Fibonacci heaps and their uses in improved network optimization
  algorithms.
\newblock {\em J. ACM}, 34(3):596--615, 1987.

\bibitem{furer-2014-set}
Martin F{\"u}rer and Huiwen Yu.
\newblock Approximating the $k$-set packing problem by local improvements.
\newblock In {\em Combinatorial Optimization}, pages 408--420, 2014.

\bibitem{garey-johnson}
Michael~R. Garey and David~S. Johnson.
\newblock {\em Computers and Intracability: {A} Guide to the Theory of
  {NP}-Completeness}.
\newblock 1979.

\bibitem{goldman-2014-spliddit}
Jonathan Goldman and Ariel~D. Procaccia.
\newblock Spliddit: unleashing fair division algorithms.
\newblock {\em SIGecom Exch.}, 13(2):41–46, 2015.

\bibitem{graham-1969-timing}
R.~L. Graham.
\newblock Bounds on multiprocessing timing anomalies.
\newblock {\em SIAM Journal on Applied Mathematics}, 17(2):416--429, 1969.

\bibitem{karp-1972-reducibility}
Richard~M. Karp.
\newblock Reducibility among combinatorial problems.
\newblock {\em Complexity of Computer Computations}, pages 85--103, 1972.

\bibitem{korf-2009-partitioning}
Richard~E. Korf.
\newblock Multi-way number partitioning.
\newblock In {\em the 21st International Joint Conferences on Artificial
  Intelligence (IJCAI)}, pages 538--543, 2009.

\bibitem{lipton-2004-envy}
R.~J. Lipton, E.~Markakis, E.~Mossel, and A.~Saberi.
\newblock On approximately fair allocations of indivisible goods.
\newblock In {\em Proceedings of the 5th ACM Conference on Electronic
  Commerce}, page 125–131, 2004.

\bibitem{procaccia-2014-fair}
Ariel~D. Procaccia and Junxing Wang.
\newblock Fair enough: guaranteeing approximate maximin shares.
\newblock In {\em Proceedings of the Fifteenth ACM Conference on Economics and
  Computation}, page 675–692, 2014.

\end{thebibliography}
\end{document}